\documentclass{article}[11pt]
\usepackage{amsmath, amssymb}
\usepackage{verbatim}
\usepackage{nicefrac}
\usepackage{graphicx,color}
\usepackage{balance}

\newtheorem{theorem}{Theorem}[section]

\newtheorem{lemma}[theorem]{Lemma}

\newtheorem{corollary}[theorem]{Corollary}

\newtheorem{observation}{Observation}

\newenvironment{proof}{\noindent {\em {Proof:}}}{$\blacksquare$\vskip \belowdisplayskip}

\newcommand{\pip}{\left(\frac{\pi(1-p)}{(1-\pi)p}\right)}
\newcommand{\qr}{\left(\frac{q}{1-q}\right)}

\advance \topmargin by -\headheight
\advance \topmargin by -\headsep
\evensidemargin \oddsidemargin
\begin{document} 

%\markboth{Bhawalkar et al.}{Value of Targeting}

\title{Value of Targeting \thanks{\{kshipra, phummel, sergeiv\}@google.com}
\thanks{A shorter version of this paper appeared in SAGT 2014}}
\author{Kshipra Bhawalkar \\ Google Inc.  
\and Patrick Hummel \\ Google Inc. 
\and Sergei Vassilvitskii\\  Google Inc.}

\maketitle
\begin{abstract}
We undertake a formal study of the value of targeting data to an advertiser.  As expected, this value is increasing in the utility difference between realizations of the targeting data and the accuracy of the data, and depends on the distribution of competing bids.  However, this value may vary non-monotonically with an advertiser's budget. Similarly, modeling the values as either private or correlated, or allowing other advertisers to also make use of the data, leads to unpredictable changes in the value of data.  We address questions related to multiple data sources, show that utility of additional data may be {\em non}-monotonic, and provide tradeoffs between the quality and the price of data sources.  In a game-theoretic setting, we show that advertisers may be worse off than if the data had not been available at all. We also ask whether a publisher can infer the value an advertiser would place on targeting data from the advertiser's bidding behavior and illustrate that this is impossible.

\end{abstract}

\section{Introduction}

Good targeting is paramount to successful advertising: showing the right ad to the right person is beneficial to all parties involved.  On the other hand, poor targeting is wasteful: knowing about the latest car is of limited use to someone who only commutes by bike; ads for umbrellas are spam to someone who lives in the desert.  In the past few years several companies have begun assisting advertisers in their targeting efforts.  Firms like BlueKai and eXelate build profiles of web users and classify them into different interest categories, such as those interested in buying a new car, traveling to Barcelona, or with an obsession over the latest gadget.  Data management is a multi billion dollar business, but there has been little analysis characterizing the value data to an advertiser. 

An advertiser usually knows which segments of the population she wants to pitch her products to. However, these may not be perfectly aligned with the classification available from the firms above. For example, a hotel owner in Milan is not explicitly interested in Americans going to Barcelona, but does know that they may respond better than the average individual to her pitch since they have already expressed willingness to travel internationally.  How much would the hotel owner be willing to pay for the ``travel to Barcelona" demographic?  On the other hand, suppose she instead appeals to people interested in fashion.  Although not all of them have expressed interest to travel, they may be tempted by a fashion week promotion that she is running.  How much should she pay for this segment?  What if she buys the two segments together? 

At first glance, the answer is relatively simple: by using the additional targeting the advertiser increases the value of the ``typical" user seeing her pitch, and thus the value for the targeting data is bounded above by the difference in the two values.  In a similar fashion, the quality of targeting, which accounts for the alignment between the segmentation provided by the data provider and that desired by the buyer also comes into play, and higher quality (better aligned) targets should garner higher prices.  However, for online advertisements, the impressions will be sold at auction, making the analysis more subtle.  We show that the actions of other bidders as well as the broader competitive landscape play a big role in determining the value of the data.  We outline our results below:

{\em Model and Basics} (Sections \ref{sec:model} and \ref{sec:databasics}). We formalize the problem and present basic findings. As expected, we prove that the value for the data is higher for data segments that occur more frequently, are closely aligned with the targeting criteria, and have a larger impact on whether the advertiser wins the auction. 

{\em Budgets} (Section \ref{sec:budgets}). We show that the addition of a budget constraint changes the nature of the optimization problem, as bidders no longer bid their true value. This leads to potentially counterintuitive behavior, where an increase in a budget results in a lower valuation for the data in an optimal solution. 

{\em Private and Correlated Values} (Section \ref{sec:correlated}). The presence of correlated values also changes the underlying optimization problem. We show that the value of data crucially depends on the competition, and one may do better trying to entice others to buy the data rather than spending the money yourself.  

{\em Game Theoretic Setting} (Section \ref{sec:multibuyers}). We characterize the nature of equilibria in a game where multiple buyers simultaneously decide whether to purchase targeting data. We show that a pure equilibrium need not exist, and advertisers can be worse off, compared to a situation where no targeting data is present.  

{\em Multiple Data Sources} (Section \ref{sec:manysignals}). We study the problem of selecting from multiple independent data sources.  We show that the value of the data is non-monotone in the number of the sources, even if they are all identical, and give a prior-independent bound on the value of data as a function of its quality. 

{\em Value to Publishers} (Section \ref{sec:reverse}). We prove that simply observing an advertiser's bids on a large number of heterogeneous impressions does not provide enough information for a publisher to infer the advertiser's value for the data. 

\subsection{Related Work}

A large body of existing work studies the effect of targeting in auctions, but most of this work is from the perspective of the publisher.  \cite{Milgrom82} addresses the question of whether a seller should enable buyers to improve targeting in a correlated values auction by revealing information about the quality of the seller's goods.  More recently, \cite{Fu12} and \cite{Hummel13} consider the question of when improved targeting  increases a seller's revenue in an auction where the buyers have independent private values.  These papers all suggest that it is often beneficial for a seller to enable buyers to more finely target in an auction.  Finally, \cite{Abraham13} analyzes how revenue is affected by the asymmetries in information possessed by different participants and finds that such asymmetries can sometimes lead to adverse revenue effects.

A related thread of analysis considers revenue optimization strategies for an auctioneer with targeting data.  \cite{Emek12} explores how an auctioneer selling a probabilistic good might reveal partial information to maximize its revenue.  \cite{Ghosh07} considers a similar problem in a discrete setting with many goods.  And \cite{Babaioff12} studies the mechanism design problem where a seller might use the asymmetry of information to maximize his revenue. All of these results describe the effects of additional data on the seller's revenue and the overall welfare. In this work we tackle the converse problem of the value of data to an individual buyer. 

There have also been papers on how targeting affects market equilibria when there are multiple publishers.  \cite{Ghosh12} considers the question of how cookie-matching affects the market equilibrium in a model in which advertisers can use a cookie from one publisher to better target users on other publishers.  And \cite{Bergemann11} investigates the question of how enabling advertisers to target certain segments of the population would affect the market equilibrium for advertising in a model of informative advertising.  This work differs from our work in that it does not consider the underlying auctions that are used to sell advertising opportunities.

Another line of work is on advertiser optimization.  Here the focus is either on getting a fair or representative allocation~\cite{RepBid}, or on finding bidding strategies that work well with poor forecasts~\cite{AdaBid2,Chen,AdaBid}. These approaches are data agnostic and do not explore advertiser actions when she can use better targeting.

Our approach is quite different from any of the above papers, as we focus on what factors affect the advertiser's utility for targeting data in an auction setting.  None of the papers considered above address this question.

\section{Preliminaries}
\label{sec:model}
We consider the setting of a single agent (advertiser) buying items (impressions) in a second price auction.  The items are heterogeneous, and different items are valued differently by the buyer.  Let $T = \{t_1, t_2, \ldots\}$ be the partition of items into types.  For an item of type $t \in T$ (\emph{e.g.}, impressions from Texas), we say that the buyer draws his value from a distribution with cdf $G_t$ and corresponding pdf $g_t$.  In the simplest case when $G_t$ is a point distribution, we denote it by $v_t$.  We denote by $\pi_t$ the {\em prior} that the buyer has on the item being of type $t$.  This is the probability that a random item is of type $t$; thus $\sum_{t} \pi_t = 1$.

In addition to knowing her prior, the buyer has access to additional data sources (signals) about the impression type.  We assume that each signal is drawn from a distribution that depends only on the impression type and that conditional on the impression type, the signals are independent from the buyer's prior and from each other.  The data sources may be imperfect. 

For part of the manuscript we will focus on the special case in which a data source only identifies whether a user is in one of two subsets of the population, which we denote by $H$ and $L$ (high and low). The advertiser then has a prior $\pi$ for the high type, and $1 - \pi$ for the low type, and values them at $v_H$ and $v_L$ respectively. In this setting we also assume that a signal has the same benefit  for predicting both types.  Formally, we model this by saying that each signal $s_i$ has a {\em quality} $q_i \in [\frac{1}{2}, 1]$, which represents the probability of the signal being correct, $q_i = Pr[s_i = t]$; the signal is then incorrect with probability $1 - q_i$.  %In addition, each signal $s_i$ has a querying cost, which we denote by $c_i$. 

The items are sold in a second price auction, and we let $f$ ($F$) denote the pdf (cdf) of the distribution of the highest competing bid.  In the independent private value (IPV) setting, the distributions are the same for all of the different types.  A natural generalization is the correlated value setting, where items of different types sell for different prices, whose densities and cumulative distribution functions we denote by $f_t$ and $F_t$ respectively.  For the rest of the paper, we will assume the IPV setting unless explicitly specified otherwise. 

We work with the standard quasilinear utility model, with buyers acting to maximize the difference between their value and the price paid. When the type is known to the buyer, she maximizes her utility by bidding her value for each impression type.  If, on the other hand, she only has the prior information about the item type, she bids $\bar{v} = \sum_t \pi_t v_t$ in the IPV setting.
\section{Data Basics}
\label{sec:databasics}
We begin by considering simple settings to develop some intuition about the value of targeting data.  We first illustrate that the buyer's value for the data results both from buying more desired items and not overpaying for lower quality items.  We extend this analysis to show that the value inherently depends on four quantities: the buyer's prior information, the quality of the signal, the difference in values for the different types, {\em and} the competitive landscape, expressed as the additional fraction of impressions the buyer can win with the value. 

\subsection{Binary User Types}
\label{sec:binary}
To develop intuition we first consider the setting in Section 2 with two types of users, $H$ and $L$.  The advertiser has access to a noisy data source that will either assume the value $h$ or $\ell$.  In particular, if the user is of type $H$ $(resp., L)$, then the data source produces a signal $h$ $(resp., \ell)$ with probability $q > \nicefrac{1}{2}$ and $\ell$ $(resp., h)$ with probability $1-q$.

In this case, if the advertiser has access to the data, she updates her prior based on the signal.  The probabilities that a user is of type $H$ upon receiving signals of $h$ and $\ell$ respectively are:
$$\pi|h = \frac{\pi q}{\pi q + (1 - \pi)(1 - q)}  \indent \indent \indent \pi|\ell = \frac{\pi(1-q)}{\pi (1-q) + (1 - \pi)q}$$
Thus the expected value of an advertising opportunity upon receiving a signal of $h$ or $\ell$ respectively is $v|h = \pi|h \cdot v_H + (1 - \pi|h)v_L$ and $v|\ell = \pi | \ell \cdot v_H + (1 - \pi | \ell) v_L$.  We now precisely characterize the value of the signal to the advertiser:

\begin{lemma}\label{l:noisy} The advertiser's value for this noisy signal is
\begin{equation*}
\int_{\bar{v}}^{v | h} \!\!\!\!(\pi q (v_H - p) + (1 - \pi)(1 - q)(v_L - p))f(p) dp  
- \int_{v | \ell}^{\bar{v}} (\pi(1 - q)(v_H - p) + (1 - \pi)q (v_L - p)) f(p) dp.
\end{equation*} 
\end{lemma}
\begin{proof} 
If the advertiser does not have access to the data source, then the advertiser obtains a utility of  
\begin{equation*}
u_{ND} = \pi\int_0^{\bar{v}} (v_H - p)f(p) dp + (1-\pi)\int_0^{\bar{v}} (v_L - p) f(p) dp.
\end{equation*}

When the advertiser has access to the data, the advertiser's utility is:
\begin{align*}
u_{D} &= \pi q \int_0^{v | h} (v_H - p) f(p) dp + \pi (1 - q) \int_0^{v| \ell} (v_H  - p) f(p) dp 
\\& +  (1 - \pi)q \int_{0}^{v | \ell}  (v_L  - p) f(p) dp  + (1 - \pi)(1 - q)  \int_{0}^{v | h} (v_L - p) f(p) dp \\
&= \pi\int_{0}^{v | \ell} (v_H - p)  f(p) dp + \pi q \int_{v | \ell}^{v | h} (v_H - p) f(p) dp +(1 - \pi)\int_{0}^{v | \ell} (v_L - p) f(p) dp  
\\& + (1 - \pi)(1 - q) \int_{v | \ell}^{v | h} (v_L - p) f(p) dp.
\end{align*}

\noindent This means that the advertiser's overall utility gain from the data for $H$ impressions is %$u_{D} - u_{ND} = $
\begin{align*}
u_D - u_{ND} &= \pi \left[  \int_{0}^{v | \ell} (v_H - p)  f(p) dp +  q \int_{v | \ell}^{v | h} (v_H - p) f(p) dp - \int_{0}^{\bar{v}} (v_H - p) f(p) dp \right]\\
%&= \pi \left[- \int_{v | \ell}^{\bar{v}} (v_H - p) f(p) dp  + q \int_{v | \ell}^{\bar{v}} (v_H - p) f(p) dp + q \int_{\bar{v}}^{v | h} (v_H - p) f(p) dp \right] \\
&= \pi \left[ q \int_{\bar{v}}^{v|h} (v_H - p) f(p) dp - (1 - q) \int_{v | \ell}^{\bar{v}} (v_H - p) f(p) dp \right].
\end{align*}

%\noindent and the advertiser's overall utility gain from the data for $L$ impressions is:
A similar analysis for the $L$ impressions gives: $u_{D} - u_{ND} = $
\begin{align*}
%u_{D} - u_{ND} &=  (1 - \pi) \left[  \int_{0}^{v | \ell} (v_L - p)  f(p) dp +  (1 - q) \int_{v | \ell}^{v | h} (v_L - p) f(p) dp - \int_{0}^{\bar{v}} (v_L - p) f(p) dp\right]\\
%&= (1 - \pi) \left[ -\int_{v | \ell}^{\bar{v}} (v_L - p) f(p) dp + \int_{v | \ell}^{\bar{v}} (1 - q) (v_L - p)f(p) dp + \int_{\bar{v}}^{v | h} (1 - q) (v_L - p) f(p) dp 
%\right]\\
&(1 - \pi)\left[ (1 - q) \int_{\bar{v}}^{v | h} (v_L - p) f(p) dp - q \int_{v | \ell}^{\bar{v}} (v_L - p) f(p) dp \right].
\end{align*}

The result follows by adding the advertiser's utility gain from the two types of impressions.
%By combining the two expressions, we see that the advertiser's total utility gain from having access to the noisy signal is :
%\begin{align*}
%u_{D} - u_{ND} = \int_{\bar{v}}^{v | h} (\pi q (v_H - p) + (1 - \pi)(1 - q)(v_L - p))f(p) dp  -  \int_{v | \ell}^{\bar{v}} (\pi(1 - q)(v_H - p) + (1 - \pi)q (v_L - p)) f(p) dp
%\end{align*}
%
\end{proof}

While Lemma \ref{l:noisy} gives an exact expression for the advertiser's value for a noisy signal, it may be not be immediately transparent how the advertiser's value for this data is affected by the different parameters in the model.  Our next result gives a simpler expression which bounds this value and shows that value can be decomposed into four independent factors. 

\begin{theorem}\label{t:noisy2}  The advertiser's value for the data can be bounded from above by $(v_H - v_L) \pi (1 - \pi) (2q - 1) (F(v | h) - F(v | \ell))$.
\end{theorem}
\begin{proof}
Consider just the first term in our expression for $u_D - u_{ND}$, and rewrite $v_H = v_L + (v_H -v_L)$.  Then we have
\begin{align*}
& \int_{\bar{v}}^{v | h} (\pi q (v_H - p) + (1 - \pi)(1 - q)(v_L - p))f(p) dp  \\
 = &\int_{\bar{v}}^{v|h} \Big(\pi q (v_H - v_L) + (v_L - p)(\pi q + (1 - \pi)(1 - q) )\Big)f(p) dp
\end{align*}

The term $\pi q (v_H - v_L) + (v_L - p)(\pi q + (1 - \pi)(1 - q))$ is a decreasing function of $p$, and it is easy to check that it equals $0$ at $p = v | h$. It achieves its maximum in the interval $[\bar{v}, v | h]$ at $p = \bar{v}$ in which case it is equal to:
\begin{align*}
\pi q (v_H - v_L) &+ (v_L - \bar{v})(\pi q + (1 - \pi)(1 - q) ) 
%&= \pi q (v_H - v_L) - \pi (v_H - v_L)( \pi q + (1 - \pi)(1 - q)) \\
%&= (v_H - v_L) \pi (q - \pi q - (1 - \pi)(1 - q)) \\
%&= (v_H - v_L) \pi (q (1 - \pi) - (1 - \pi)(1 - q)) \\
= (v_H - v_L) \pi (1 - \pi)(2q - 1)
\end{align*}

Therefore the maximum value of the integral is: $$(v_H - v_L) \pi (1 - \pi)(2q - 1)\Big(F(v|h) - F(\bar{v})\Big).$$
The analysis for the other integral is similar,  therefore the total value of the data is bounded by:
$$(v_H - v_L) \pi (1 - \pi) (2q - 1) \Big(F(v | h) - F(v | \ell)\Big). \qquad $$
\end{proof}

Theorem \ref{t:noisy2} provides some basic rules of thumb for valuing the data.  If there is not much difference between the advertiser's values for advertising to different types of users, $(v_H - v_L)$, then the advertiser will not care much whom she advertises to, and will have little value for the data.  Thus the advertiser's value for the data is increasing in this quantity.  It is also intuitive that the advertiser's value for the data is increasing in the accuracy of the signal, $(2q - 1)$.

Additionally, if  competing buyers rarely place a bid that falls between the values the advertiser may have for the different types of users, an advertiser's ability to adjust her bid in response to the different possible realizations of the targeting data will rarely have an effect on whether she wins the advertising opportunity.  Thus the value of the data is increasing in the likelihood of a competing advertiser placing a bid between the advertiser's possible values for the different types of advertising opportunities, $(F(v|h) - F(v|\ell))$. 

Finally, it makes sense that the advertiser's value for the data is single-peaked in $\pi$.  If $\pi$ is very close $0$ or $1$, then the data almost always takes on the same value, and there is little gain to seeing it.  By contrast, when $\pi$ is closer to $\nicefrac{1}{2}$, there is more uncertainty in the true type of the item, and thus more heterogeneity in the different realizations of the targeting data, so the data is more valuable.

\subsection{General Distributions of Valuations}
\label{sec:general}
We now move to a more general model, where we consider what happens when an advertiser's best estimate for his value for advertising to a particular user may assume a large number of distinct values. We show formally that more ``refined'' signals on the value of the item are more valuable to the advertiser.  

To model this setting, suppose that when an advertiser has access to data, she learns that the best estimate for her value for advertising to a particular user is $v$, which is a random draw from some distribution $G(\cdot)$ with corresponding pdf $g(\cdot)$.  If an advertiser has no access to data, then the advertiser simply knows her expected value,  $E[v | v \thicksim G] = \bar{v}$. The utility gain from the data is then $\int_0^{\infty} \int_{\bar{v}}^v (v-p) f(p) dp \; g(v) dv$. 

Now suppose there are two different data sources that an advertiser might use, with distributions $G$ and $H$. If they are unbiased, the distributions will satisfy $E_G[v] = E_H[v]$. We address the question of when one data source would be more useful to an advertiser than another:

\begin{theorem}\label{t:sosd}  Consider two data sources with corresponding cdfs $G$ and $H$ satisfying $E_G[v] = E_H[v]$. Then if $H(\cdot)$ second order stochastically dominates $G(\cdot)$, the advertiser has more value for the data source $G$ than the data source $H$. 

\end{theorem}
\begin{proof}
For a given data source $S \in \{G, H\}$, the advertiser's utility gains can be written as: 
%
%
%Note that the advertiser's utility gains from using the data sources $G$ and $H$ respectively are as follows:
%
%\begin{equation*}
%u_G = \int_0^{\infty} \int_{\bar{v}}^v (v-p) f(p) dp \; dG(v)
%\end{equation*}
 %
\begin{equation*}
u_S = \int_0^{\infty} \int_{\bar{v}}^v (v-p) f(p) dp \; dS(v).
\end{equation*}
Now note that $u(v) \equiv \int_{\bar{v}}^v (v-p) f(p) dp$ is a convex function of $v$ because 
$$\frac{d^2}{dv^2} \int_{\bar{v}}^v (v-p) f(p) dp = \frac{d}{dv} \int_{\bar{v}}^v f(p) dp = \frac{d}{dv} [F(v) - F(\bar{v})] = f(v) \geq 0.$$  
We can thus rewrite $u_S$ as $u_S = \int_0^{\infty} u(v) \; dS(v)$ for some convex function $u(v)$.

Now consider what happens if $H(\cdot)$ second order stochastically dominates $G(\cdot)$.  Since $G(\cdot)$ is a mean-preserving spread of $H(\cdot)$, it follows that
\begin{equation*}
u_G = \int_0^{\infty} u(v) \; dG(v) = \int_0^{\infty} E[u(w + \epsilon_w) | w] \; dH(w)
\end{equation*}
\noindent for some distributions $\epsilon_w$ (that may depend on $w$) satisfying $E[\epsilon_w | w] = 0$ for all $w$.  Thus
\begin{align*}
u_G  &= \int_0^{\infty} E[u(w + \epsilon_w) | w] \; dH(w) \\&\geq \int_0^{\infty} u(E[w + \epsilon_w | w]) \; dH(w) = \int_0^{\infty} u(w) \; dH(w) = u_H,
\end{align*}
\noindent where the inequality follows from Jensen's inequality.
\end{proof}

\subsection{Budgets}
\label{sec:budgets}
Throughout the analysis so far we have assumed that an advertiser does not face any budget constraints, but in some settings an advertiser only has a fixed budget for advertising.  It is natural to ask how this possibility would affect an advertiser's value for targeting data.  We address this possibility in this section.

We again consider the binary case, with types $H$ and $L$ and corresponding price distributions $F_H(\cdot)$ and $F_L(\cdot)$.  Let $f(\cdot) \equiv \pi f_H(\cdot) + (1-\pi) f_L(\cdot)$ be the density of the highest competing bid unconditional on type.  In addition, let $B$ denote the maximum amount the advertiser can spend per advertising opportunity.  Note that if the advertiser does not have access to the data, then in every auction she makes a bid $b$ satisfying $\int_{0}^{b} p f(p) dp = B$.  

If the advertiser has access to a perfectly informative data source, then the advertiser bids $b_L (b_H)$ when the user is type $L(H)$, where the values of $b_L$ and $b_H$ must satisfy $\pi \int_{0}^{b_H} p f_H(p) dp + (1-\pi) \int_{0}^{b_L} p f_L(p) dp = B$.  The advertiser then chooses the values of $b_L$ and $b_H$ to maximize $\pi \int_{0}^{b_H} v_H f_H(p) dp + (1-\pi) \int_{0}^{b_L} v_L f_L(p) dp - B$ subject to this budget constraint. 

\begin{lemma}\label{t:budget}  Suppose the advertiser faces a binding budget constraint.  Then if the advertiser bids $b_H (b_L)$ for users of type $H (L)$, it must be the case that $\frac{b_H}{b_L} = \frac{v_H}{v_L}$.
\end{lemma}
\begin{proof} 
The advertiser chooses values of $b_L$ and $b_H$ that maximize the Lagrangian,
\begin{equation*}
\pi \int_{0}^{b_H} (v_H - \lambda p) f_H(p) dp + (1-\pi) \int_{0}^{b_L} (v_L - \lambda p) f_L(p) dp,
\end{equation*}
\noindent for some value of $\lambda$.  Differentiating with respect to $b_H$ and $b_L$ respectively indicates that the advertiser chooses values of $b_L$ and $b_H$ that satisfy $\pi (v_H - \lambda b_H) f_H(b_H) = 0$ and $(1-\pi) (v_L - \lambda b_L) f_L(b_L) = 0$, meaning the advertiser chooses values of $b_L$ and $b_H$ that satisfy $v_H = \lambda b_H$ and $v_L = \lambda b_L$.  From this it follows that $\frac{b_H}{b_L} = \frac{v_H}{v_L}$.  
\end{proof}

This result immediately implies that a budget-constrained advertiser's value for the data is

\begin{equation*}
u_D - u_{ND} = \pi v_H F_H \left(\frac{v_H b_L}{v_L}\right) + (1-\pi) v_L F_L(b_L) - \overline{v} F(b).
\end{equation*}

Now we address how small changes in the budget (which induce small changes in $b$, $b_L$, and $b_H$)  affect the advertiser's value for the data.  One might think that an advertiser's value for the data would always increase with the advertiser's budget because the advertiser would be better able to exploit the targeting information.  However, this is not the case as the following theorem illustrates:

\begin{theorem}\label{t:budget2}  The advertiser's value for the data is not monotone in her budget.

\end{theorem}
\begin{proof}
First note that it is possible for the advertiser's value for the data to increase as a result of a small change in the advertiser's budget.  If the advertiser initially has a budget of zero, then she also has zero value for the data.  But if the advertiser has a small positive budget, then the advertiser has a strictly positive value for the data.

Next we show that it is possible for the advertiser's value for the data to {\em decrease} as a result of small increases in the advertiser's budget.  To see this, consider the special case where $\pi = \frac{1}{2}$ and $f_H(\cdot) = f_L(\cdot) = f(\cdot)$, where $f(p) = 2$ if $p \in [0, \frac{1}{2})$ and $f(p) = 0$ otherwise.  Consider values of $B$ such that the optimal values of $b$ and $b_L$ are both very close to $\frac{1}{2}$ (and $b_H > \frac{1}{2}$).  We then have:
%Given the general expression of $u_D - u_{ND}$, we have
\begin{align*}
%
%\begin{equation*}
\frac{d}{db} [u_D - u_{ND}] 
 = \left[\pi \frac{v_H^2}{v_L} f_H \left(\frac{v_H b_L}{v_L}\right) + (1-\pi) v_L f_L(b_L) \right]\frac{db_L}{db} - \overline{v} f(b) 
%\end{equation*}
%
%\noindent so in this special case we have
%
%\begin{equation*}
%\\&=\frac{d}{db} [u_D - u_{ND}] = 
=\frac{1}{2} v_L f(b_L) \frac{db_L}{db} - \overline{v} f(b)
%\end{equation*}
\end{align*}
Next note that
\begin{equation*}
\int_{0}^{b} p f(p) dp = \pi \int_{0}^{\frac{v_H b_L}{v_L}} p f_H(p) dp + (1-\pi) \int_{0}^{b_L} p f_L(p) dp.
\end{equation*}
Differentiating the left-hand side of this equation with respect to $b$ gives $b f(b)$ and differentiating the right-hand side of this equation with respect to $b_L$ gives $\pi (\frac{v_H}{v_L})^2 b_L f_H(\frac{v_H b_L}{v_L}) + (1-\pi) b_L f_L(b_L)$.  From this it follows that
\begin{align*}
\frac{db_L}{db} &= \frac{b f(b)}{\pi (\frac{v_H}{v_L})^2 b_L f_H(\frac{v_H b_L}{v_L}) + (1-\pi) b_L f_L(b_L)}
%\end{equation*}
%
%%\noindent Thus if we again consider the special case that has been the focus of this proof, we find that
%
%\begin{equation*}
 =\frac{b f(b)}{\frac{1}{2} b_L f(b_L)}
\end{align*}
By combining the expressions we have derived for $\frac{d}{db} [u_D - u_{ND}]$ and $\frac{db_L}{db}$, we see that
\begin{equation*}
\frac{d}{db} [u_D - u_{ND}] = \left(\frac{v_L b}{b_L} - \overline{v}\right) f(b).
\end{equation*}
But in the limit as $b$ and $b_L$ become arbitrarily close to $\frac{1}{2}$, the above expression approaches $(v_L - \overline{v}) f(b) < 0$.  Thus it is possible to have $\frac{d}{db} [u_D - u_{ND}] < 0$, and the advertiser's value for the data may decrease as a result of increases in the advertiser's budget.  
\end{proof}

While it is intuitive that an advertiser's value for the data may increase with the advertiser's budget, it may be less obvious why an advertiser's value for the data might decrease with the advertiser's budget.  To see why this might arise, suppose the advertiser always has a larger value for all advertising opportunities than the competing advertisers.  In this case, if an advertiser has a large budget, the data hardly has any effect on the impressions that the advertiser purchases since she would purchase almost all impressions anyway.  However, if the advertiser has a smaller budget, then the targeting data may have a significant effect on which advertising opportunities the advertiser wins.  Thus the advertiser's value for the data may be decreasing in the size of the advertiser's budget.

\subsection{Correlated Value Setting}
\label{sec:correlated}

In Sections 3.1 and 3.2 we restricted attention to scenarios in which there is no correlation between an advertiser's value for an advertising opportunity and the highest competing bid.  However, in practice this assumption does not always hold because user characteristics like income make advertising opportunities more or less valuable to multiple advertisers at the same time. In this section, we address how this possibility affects an advertiser's value for targeting data.

To illustrate our results, we again consider the binary case, with users of type $H$ and $L$.  However, we now assume that if the user is of type $H (resp., L)$, then the highest bid placed by a competing advertiser is a random draw from the distribution $F_H(\cdot) (resp., F_L(\cdot))$ with pdf $f_H(\cdot) (resp., f_L(\cdot))$.

Before we can figure out how the advertiser would value targeting data, we must first figure out the bidding strategy that the advertiser would use when she does not have access to this targeting data.  

\begin{lemma}
With continuous densities, in equilibrium, the advertiser will bid some amount $v^* \in [v_L, v_H]$ that satisfies
\begin{equation}\label{e:correlated-value-bid}
v^* = \frac{\pi f_H(v^*) v_H + (1-\pi) f_L(v^*) v_L }{\pi f_H(v^*) + (1-\pi) f_L(v^*)}.
\end{equation}
\end{lemma}
\begin{proof}
Note that, conditional on the highest bid from competing advertisers being $v^*$, the probability the user is type $H$ is

\begin{equation*}
\frac{\pi f_H(v^*)}{\pi f_H(v^*) + (1-\pi) f_L(v^*)}
\end{equation*}

\noindent Thus if the highest bid from competing advertisers is $v^*$, then an advertiser's expected value for the advertising opportunity is

\begin{equation*}
\frac{\pi f_H(v^*) v_H + (1-\pi) f_L(v^*) v_L }{\pi f_H(v^*) + (1-\pi) f_L(v^*)}
\end{equation*}

Now when an advertiser decides whether to make a small deviation from bidding $v^*$, the advertiser can condition on the possibility that the highest bid from competing advertisers is also $v^*$, as a small change in the advertiser's bid will not have any effect on whether the advertiser wins the auction otherwise.  Thus in order for a bid to be optimal, it must be the case that the bid is equal to the advertiser's expected value for winning an advertising opportunity conditional on the highest bid from competing advertisers also being equal to this bid.  From this it follows that the advertiser must follow a strategy of bidding some amount $v^*$ that satisfies the equation

\begin{equation*}
v^* = \frac{\pi f_H(v^*) v_H + (1-\pi) f_L(v^*) v_L }{\pi f_H(v^*) + (1-\pi) f_L(v^*)}
\end{equation*}

\noindent when the advertiser does not have access to the data.

To see that there is some $v^* \in [v_L, v_H]$ that satisfies this equation, note that the right-hand side of the above equation is always greater than or equal to $v_L$ and less than or equal to $v_H$.  Thus when $v^* = v_L$, the left-hand side of this equation is less than or equal to the right-hand side, and when $v^* = v_H$, the left-hand side of this equation is greater than or equal to the right-hand side.  By the intermediate value theorem, it then follows that there exists some $v^*$ that satisfies this equation. 
\end{proof}

\noindent Thus if the advertiser does not have access to the data,  her expected payoff is $u_{ND} = \pi \int_0^{v^*} (v_H - p) f_H(p) dp + (1-\pi)\int_0^{v^*} (v_L - p) f_L(p) dp$.  However, if the advertiser has access to the data, then she places a bid of $v_H (v_L)$ for users of type $H (L)$.  The advertiser's expected payoff is then $u_D = \pi \int_0^{v_H} (v_H - p) f_H(p)dp + (1-\pi)\int_0^{v_L} (v_L - p) f_L(p) dp$.  Her value for the data is the difference between the two expressions and it depends on similar things in the correlated value setting as in the private value setting.  In particular, we obtain the following result:

\begin{theorem}\label{t:correlated-value}  An advertiser's value for targeting data is increasing in her utility difference between advertising to different users ($v_H - v_L$), increasing in the likelihood of a competing advertiser placing a bid between these possible values ($f_H(p)$ for $p \in [v^*, v_H]$ and $f_L(p)$ for $p \in [v_L, v^*]$), and is single-peaked in the relative likelihoods of the different realizations of the targeting data ($\pi$).

\end{theorem}
\begin{proof}
Note that
\begin{align*}
\frac{d}{dv_H} [u_{D} - u_{ND}] &= 
%pi \int_{v^*}^{v_H} f_H(p) dp - \pi \frac{dv^*}{dv_H} (v_H - v^*)f_H(v^*) + (1-\pi)\frac{dv^*}{dv_H}(v^* - v_L) f_L(v^*) \\
 \pi \int_{v^*}^{v_H} f_H(p) dp > 0\\
%\end{align*}
%
%\noindent Also note that
%
%\begin{align*}
\frac{d}{dv_L} [u_{D} - u_{ND}] &= %-(1-\pi) \int_{v_L}^{v^*} f_L(p) dp - \pi \frac{dv^*}{dv_L} (v_H - v^*)f_H(v^*) + (1-\pi) \frac{dv^*}{dv_L} (v^* - v_L) f_L(v^*) \\
 -(1-\pi) \int_{v_L}^{v^*} f_L(p) dp < 0
\end{align*}

Thus the advertiser's value for the data is increasing in $v_H$ and decreasing in $v_L$, meaning the advertiser's value for the data is also increasing in $v_H - v_L$.

Also note that the expression for $u_{D} - u_{ND}$ is increasing in $f_H(p)$ for all $p \in [v^*, v_H]$ and increasing in $f_L(p)$ for all $p \in [v_L, v^*]$.  Thus the advertiser's value for the data is also increasing in $f_H(p)$ for all $p \in [v^*, v_H]$ and increasing in $f_L(p)$ for all $p \in [v_L, v^*]$.

Finally note that
\begin{align*}
\frac{d}{d\pi} [u_{D} - u_{ND}] &= % \int_{v^*}^{v_H} (v_H - p) f_H(p) dp - \int_{v_L}^{v^*} (p - v_L) f_L(p) dp \\ & - \pi(v_H - v^*)f_H(v^*) \frac{dv^*}{d\pi}  + (1-\pi)(v^* - v_L)f_L(v^*)\frac{dv^*}{d\pi} \\
 \int_{v^*}^{v_H} (v_H - p) f_H(p) dp - \int_{v_L}^{v^*} (p - v_L) f_L(p) dp
\end{align*}
This expression tends to $\int_{v_L}^{v_H} (v_H - p) f_H(p) dp > 0$ in the limit as $\pi \to 0$ and $v^* \to v_L$ and tends to $-\int_{v_L}^{v_H} (p - v_L) f_L(p) dp < 0$ in the limit as $\pi \to 1$ and $v^* \to v_H$.  The expression is also strictly decreasing in $\pi$ since $v^*$ is strictly increasing in $\pi$, the first integral is strictly decreasing in $v^*$, and the second integral is strictly increasing in $v^*$.  Thus there is some $\pi^{*} \in (0, 1)$ such that $\frac{d}{d\pi} [u_{D} - u_{ND}] > 0$ for $\pi < \pi^{*}$ and $\frac{d}{d\pi} [u_{D} - u_{ND}] < 0$ for $\pi > \pi^{*}$.  From this it follows that the advertiser's value for the data, $u_{D} - u_{ND}$, is single-peaked in $\pi$.  
\end{proof}

While the advertiser's values for the data in the correlated and private value settings depend on similar terms, the two values are incomparable: the value is not always higher in one setting than the other.  In the correlated value framework, it is entirely possible for the advertiser to have zero value for the data, as the advertiser may be able to exploit the fact that the competing advertisers are perfectly segmenting the market, so that she always wins high value impressions  and always loses on the low value impressions.  In these circumstances, the advertiser's value for the data will be lower under the correlated value framework than in the private value framework.

At the same time, it is also possible that the advertiser could have greater value for the data in the correlated value framework than in the private value framework.  If the competing advertisers are making bids that are strongly correlated with the advertiser's value for the advertising opportunity, then the advertiser may not be able to profitably bid in the auction without access to the data.  In this case, the data is especially valuable for the advertiser in the correlated value framework, and the data may be more valuable under correlated values than under private values. We make this precise below. 

\begin{observation}\label{t:correlated-value2}  The advertiser's value for the data may be either greater or lower in the correlated value setting than in the private value setting.

\end{observation}
\begin{proof} 
First note that it is possible for the advertiser's value for the data to be lower in the correlated value setting than it is in the private value setting.  To see this, suppose that $\pi = \frac{1}{2}$, $v_L = 0$, $v_H = 1$, $f_L(p) = 2$ if $p \in [\frac{1}{2}, 1]$, $f_L(p) = 0$ otherwise, $f_H(p) = 2$ if $p \in [0, \frac{1}{2})$, and $f_H(p) = 0$ otherwise.  Then if the distribution of the highest bid from competing advertisers is a distribution with pdf $f(p) = \pi f_H(p) + (1-\pi)f_L(p) = 1$ (for $p \in [0, 1]$) regardless of the user's type, the advertiser's utility gain from having access to data is
\begin{equation*}
u_{D} - u_{ND} = \frac{1}{2} \int_{1/2}^{1} (1 - p) dp  + \frac{1}{2} \int_{0}^{1/2} p \; dp.
\end{equation*}
But if the distribution of the highest bid from competing advertisers comes from the distribution $F_H(\cdot)$ if the user is type $H$ and from the distribution $F_L(\cdot)$ if the user is type $L$, then the advertiser's utility gain from having access to the data is zero.  In this case, the best scenario for the advertiser would be to win all type $H$ impressions and to not win any type $L$ impressions.  But the advertiser can guarantee this without the data by simply always placing a bid of $b = \frac{1}{2}$.  Thus in this case the  advertiser's value for the data is zero, and it follows that the advertiser's value for the data is lower in the correlated value setting than it is in the private value setting.

Next, we show that it is possible for the advertiser's value for the data to be greater in the correlated value setting than it is in the private value setting.  To see this, suppose that $\pi = \frac{1}{2}$, $v_L = 0$, $v_H = 1$, $f_L(p) = 2$ if $p \in [0, \frac{1}{2})$, $f_L(p) = 0$ otherwise, $f_H(p) = 2$ if $p \in [\frac{1}{2}, 1]$, and $f_H(p) = 0$ otherwise.  Then if the highest bid from competing advertisers comes from a distribution with pdf $f(p) = \pi f_H(p) + (1-\pi)f_L(p) = 1$ (for $p \in [0, 1]$) regardless of the user's type, the advertiser's utility gain from having access to the data is again
\begin{equation*}
u_{D} - u_{ND} = \frac{1}{2} \int_{1/2}^{1} (1 - p) dp  + \frac{1}{2} \int_{0}^{1/2} p \; dp.
\end{equation*}
But if the highest bid from competing advertisers comes from the distribution $F_H(\cdot)$ if the user is type $H$ and from the distribution $F_L(\cdot)$ if the user is type $L$, then the advertiser's utility gain from having access to the data is
\begin{equation*}
u_{D} - u_{ND} = \frac{1}{2} \int_{1/2}^{1} (1 - p) \; 2 \; dp  + \frac{1}{2} \int_{0}^{1/2} p \; 2 \; dp,
\end{equation*}
\noindent which is strictly greater than the gain in the IPV model. Therefore, the advertiser's value for the data can also be greater in the correlated value setting than it is in the private value setting.  This proves the result. 
\end{proof}

\section{Game Theoretic Setting} \label{sec:multibuyers}

So far we have focused on the value a particular advertiser would place on targeting data while ignoring the possibility that other advertisers may also use this data.  We now ask how this possibility affects an individual advertiser's value for the data in a game-theoretic setting. 

\subsection{Value of Data}

We show that when there are two bidders, each bidder prefers the competitor to buy the data {\em regardless} of her own actions. However, the exact value an advertiser places on the data may go up or down as a function of the other bidder's actions, and there are situations where there are no pure strategy equilibria to the game where advertisers simultaneously decide whether to purchase the data. 

Consider the case with two advertisers. Each advertiser $i \in \{1,2\}$ has a value $v_{i}$ for a given advertising opportunity, where each $v_{i}$ is an independent draw from the distribution $F_{i}(v)$ with pdf $f_{i}(v)$.  If advertiser $i$ has access to the data, then advertiser $i$ knows her value for a particular advertising opportunity.  If not, then this advertiser only knows the distribution from which her value is drawn.  Naturally each advertiser obtains positive value from having access to the data regardless of whether the other advertiser also has access to the data.  Less obviously, each advertiser also has preferences over whether the other advertiser has access to the data.  In particular, we obtain the following result:

\begin{theorem}\label{t:two-buyers}  When there are two advertisers, each advertiser prefers that the other advertiser have access to the data regardless of whether she has access to the data herself.

\end{theorem}
\begin{proof}
Throughout this proof we let $\overline{v_{i}} \equiv E_{F_i}[v] $ and assume without loss of generality that $\overline{v_{1}} \geq \overline{v_{2}}$.  We also let $\delta_{i} \in \{0, 1\}$ be an indicator variable representing whether advertiser $i$ has access to the data.  Finally let $u_{i}(\delta_{1}, \delta_{2})$ denote advertiser $i$'s expected utility for given values of $\delta_{1}$ and $\delta_{2}$.

First note that
\begin{equation*}
u_{2}(1,0) = \int_{0}^{\overline{v_{2}}} (\overline{v_{2}} - v_{1}) f_{1}(v_{1}) dv_{1}  > 0 = u_{2}(0,0)
\end{equation*}
\noindent Also note that $u_{1}(0,0) = $
\begin{equation*}
%u_{1}(0,0) = 
\int_{0}^{\infty} (\overline{v_{1}} - v_{2}) f_{2}(v_{2}) dv_{2} < \int_{0}^{\overline{v_{1}}} (\overline{v_{1}} - v_{2}) f_{2}(v_{2}) dv_{2} = u_{1}(0,1)
\end{equation*}
\noindent and $u_1(1,0) = $
\begin{align*}
 \int_{0}^{\infty} \int_{\overline{v_{2}}}^{\infty} (v_{1} - v_{2}) f_{2}(v_{2}) f_{1}(v_{1}) dv_{1} dv_{2} 
< \int_{0}^{\infty} \int_{v_{2}}^{\infty} (v_{1} - v_{2}) f_{2}(v_{2}) f_{1}(v_{1}) dv_{1} dv_{2} = u_{1}(1,1)
\end{align*}
Similarly it follows that $u_{2}(0,1) < u_{2}(1,1)$.  Thus each advertiser prefers that the other advertiser have access to the data regardless of whether she has access herself. 
\end{proof}

To understand the intuition behind this result, first note that when neither advertiser has access to the data, then some advertiser (say advertiser $2$) never wins the auction. However, if advertiser $1$ has access to the data, then she may sometimes bid less than advertiser $2$ is bidding, which leads to advertiser $2$  earning positive profit.  Similarly, if advertiser $2$ has access to the data, then advertiser $1$ will sometimes pay less for advertising opportunities that she would have won anyway (because advertiser $2$ may discover that she values these impressions for less than she originally thought), but advertiser $1$ will not have to pay for impressions where advertiser $2$ learns that she values these impressions for more than advertiser $1$.  Thus advertiser $1$ is also better off when advertiser $2$ has access to the data.

This result does not necessarily extend when there are more than two advertisers.  Nonetheless, it does illustrate that an advertiser can have preferences over whether competing advertisers have access to targeting data, and that these preferences may be the opposite of what one might conjecture naively.

An advertiser's value for the data can also depend on whether the other advertiser has access to the data.  However, there is no general relationship as to how an advertiser's value for the data depends on whether the other advertiser has access to the data, even with two bidders. This is a corollary of Theorem \ref{t:correlated-value}: by purchasing the data the second advertiser may affect the competing bid distribution in an arbitrary manner, thereby changing the value of the data.

\begin{observation}\label{t:two-buyers2}  An advertiser's value for the data may either increase or decrease as a result of the other bidder having access to the data in both the private and correlated value framework.
\end{observation}
\begin{proof}
First we show that an advertiser can have a greater value for the data if the competing advertiser has access to the data than if this advertiser does not have access to the data.  To see this, suppose that $v_{1}$ is either equal to $4$ or $8$, each with probability $\frac{1}{2}$, and $v_{2}$ is either equal to $2$ or $5$, again each with probability $\frac{1}{2}$.  In this case, if advertiser $1$ does not have access to the data, then advertiser $1$ always makes a bid of $6$, and advertiser $2$ would never want to bid more than advertiser $1$ even if advertiser $2$ knew that she had a value of $5$.  Thus in this case, advertiser $2$ has no value for the data if advertiser $1$ does not have access to the data.  However, if advertiser $1$ had access to the data, then advertiser $1$ would make a bid of $4$ in cases where advertiser $1$ learned that $v_{1} = 4$.  Advertiser $2$ would then want to win this auction if she knew that she had a value of $v_{2} = 5$ but not if she knew that $v_{2} = 2$, so advertiser $2$ would have a positive value for the data in this case.  This example illustrates that an advertiser can have a greater value for the data if the competing advertiser has access to the data than if this advertiser does not have access to the data. 

To see that an advertiser may have less value for the data if the competing advertiser has access to the data than if this competing advertiser did not have access to the data, suppose that $v_{1}$ is either equal to $0$ or $8$, each with probability $\frac{1}{2}$, and $v_{2}$ is either equal to $2$ or $5$, again each with probability $\frac{1}{2}$.  In this case, if advertiser $1$ does not have access to the data, then advertiser $1$ always makes a bid of $4$.  Advertiser $2$ would then want to win this auction if she knew that he had a value of $v_{2} = 5$ but not if she knew that $v_{2} = 2$, so advertiser $2$ would have a positive value for the data in this case.  However, if advertiser $1$ did have access to the data, then advertiser $1$ would either make a bid of $8$ or a bid of $0$.  In the first case, advertiser $2$ would never want to bid more than advertiser $1$ even if advertiser $2$ knew that she had a value of $5$, and in the second case, advertiser $2$ would always want to bid more than advertiser $1$ and win the advertising opportunity for a price of $0$ even if advertiser $2$ knew that he only had a value of $2$, so advertiser $2$ would have no value for the data in this case.  This example illustrates that an advertiser can have less value for the data if the competing advertiser has access to the data than if this competing advertiser did not have access to the data.  This proves the result for the private value framework.  The result for correlated values follows the analysis used to prove Theorem \ref{t:no-pure-eq}. 
\end{proof}

\subsection{Data Buying Equilibria}
We now turn to the question of whether there exists a pure strategy equilibrium to the game in which advertisers simultaneously decide whether to purchase data and then bid in the auction.  In general the fact that one advertiser may prefer everyone to have the data, whereas another may prefer to be the unique holder of the data, means that pure strategy equilibria to this game need not exist.  

\begin{theorem}\label{t:no-pure-eq}  There may be no pure strategy equilibrium to the game in which advertisers simultaneously decide whether to purchase a data source.
\end{theorem}
\begin{proof}
Suppose there are three advertisers, Alice, Bella and Candice. Bella's value, $v_2$, is a random draw from the exponential distribution with mean $1$, and Alice's value is double that, $v_1 = 2 v_2$.  Candice's value, $v_3$, is $0$ with probability $\nicefrac{1}{10}$ independently of the values of $v_1$ and $v_2$ and is $v_1 + 5$ otherwise. Candice always knows her exact value, but Alice and Bella learn the exact realizations of their values only if they purchase the data.  If Bella does not purchase the data, then she only knows the distribution from which her value is drawn.  If Alice does not purchase the data, then she knows her value on some small random fraction of the available advertising opportunities independent of the advertisers' values, but she only knows the distribution from which her value is drawn for all other advertising opportunities.

We now examine the strategies for the three advertisers. 
First note that if Bella does not buy the data, then she has a dominant strategy of bidding $0$ when she does not know her value: the only way for Bella to win is by bidding more than Alice, and thereby bidding more than her value, so she is content with bidding $0$ in these cases.  In this case Alice's best response is to also not buy the data since she wins all of the impressions where $v_3 = 0$ at zero cost and she will not win any of the other impressions even with the data. 

But if Alice does not buy the data and the data is sufficiently cheap, then Bella's best response is to buy the data since buying the data enables her to outbid advertisers $1$ and $3$ on some high value impressions, achieving a positive profit in expectation.  However, in this case, Alice's best response is to buy the data so that she can once again win all of the advertising opportunities where $v_3 = 0$. Hence there is no pure strategy equilibrium in this game when advertisers simultaneously decide whether to purchase the data.
\end{proof}

While it is generally the case that there will not exist a pure strategy equilibrium to the game in which advertisers simultaneously decide whether to purchase a given source of data, there are important special cases under which such pure strategy equilibria exist.  In particular, in a pure private values setting with symmetric bidders, we can prove that pure strategy equilibria exist. 

\begin{theorem}\label{t:pure-eq}  Suppose each advertiser's estimate of her value is an independent draw from the distribution $H(\cdot)$ if the advertiser purchases the data and from the distribution $G(\cdot)$ otherwise.  Then there is a pure strategy equilibrium to the game in which advertisers simultaneously decide whether to purchase the data.
\end{theorem}
\begin{proof}
Let $n$ denote the number of advertisers and suppose by means of contradiction that there does not exist a pure strategy equilibrium to the game in which advertisers simultaneously decide whether to purchase the data. Let $k$ denote the smallest positive integer for which if exactly $k$ of the advertisers choose to purchase the data, then one of the advertisers can profitably deviate by not purchasing the data. Note that some such $k \leq n$ must exist since all $n$ buyers purchasing the data is not an equilibrium by assumption. 

Thus, if exactly $k-1$ advertisers purchase the data, then no advertiser who purchases the data can profitably deviate by not purchasing the data. If this is not an equilibrium, then an advertiser who doesn't purchase the data can profitably deviate by purchasing the data. But this contradicts the definition of $k$ because this implies that if $k$ advertisers purchase the data, then these advertisers indeed prefer to be purchasing the data than to not be purchasing the data. Therefore, there must exist a pure strategy equilibrium to the game in which advertisers simultaneously decide whether to purchase the data.
\end{proof}

Finally, we present a result on the value of data when multiple symmetric bidders are all given access to the same data. In general, sources of data may refine an advertiser's estimate of her value by helping an advertiser learn that her true value for an advertising opportunity is $v + \epsilon$  for some mean-zero random variable $\epsilon$, where $v$ denotes the advertiser's original estimate of her value.  It is interesting to ask how the value of a data source depends on whether it helps advertisers distinguish amongst high or low value advertising opportunities.  We show below that the data is more valuable for picking out high valued impressions when multiple advertisers are given access to the same data.

To illustrate this, we consider a situation in which each advertiser's initial estimate of her value is an independent and identically distributed draw from the distribution $F(\cdot)$. We call a data source a {\em $(v^*, \delta)$-refinement} signal, if, given an initial estimate of an advertiser's value $v$:
\begin{itemize}
\item it provides no information for $v \not \in [v^* - \delta, v^*+\delta]$,
\item refines the value to  $v + \epsilon$ when $v \in [v^*-\delta, v^*+\delta]$ for some $\epsilon$ that is a random draw from some distribution with mean zero and narrow support.  %We refer to such a data source as a data source that refines an advertiser's estimate of her value for values near $v^*$.  For such a data source, we obtain the following result:
\end{itemize}

\begin{theorem}
\label{t:sym-eq}  Suppose there are at least three advertisers, and all advertisers are given access to a $(v^*, \delta)$ refinement signal for some $v^*$ in the support of $F(\cdot)$ and an arbitrarily small value of $\delta$.  Then there is a threshold $\tau$ in the interior of the support of $F(\cdot)$ such that the value of the data source to the advertisers is positive for values of $v^* > \tau$ and negative for values of $v^* < \tau$.
\end{theorem}
\begin{proof}
Let $n$ denote the number of advertisers. Note that if only one advertiser has an initial estimate close to $v^*$ then the data source will affect neither the winner nor the payment made by the winner in expectation.  Thus the data source can only affect the bidders' expected payoffs if either the two highest bidders both have initial estimates of their values that are arbitrarily close to $v^*$ or the second and third highest bidders both have initial estimates of their values that are arbitrarily close to $v^*$.

Now if the two highest bidders both have initial estimates of their values that are arbitrarily close to $v^*$, then after the refinement, in expectation the difference in their values increases, and thus the data increases the bidders' expected payoffs.  If the second and third highest bidders both have initial estimates of their values that are arbitrarily close to $v^*$, then in expectation the maximum of the two bids increases after refinement, and thus the data decreases the winning bidder's expected payoff in expectation.  The relative likelihood that the two highest bidders will both have initial estimates of their values that are arbitrarily close to $v^*$ is ${n \choose 2} f(v^*)^2 F(v^*)^{n-2}$, and the relative likelihood that the second and third highest bidders will both have initial estimates of their values that are arbitrarily close to $v^*$ is $\frac{n!}{(n-3)! 2} f(v^*)^2 (1-F(v^*)) F(v^*)^{n-3}$.  Thus the ratio of these likelihoods is $\frac{F(v^*)}{(n-2)(1-F(v^*))}$, which is increasing in $v^*$, becomes arbitrarily large as $v^*$ approaches the supremum of the support of $F(\cdot)$, and tends to $0$ as $v^*$ approaches the infinum of the support of $F$.  Therefore there exists a $\tau$ in the interior of the support of $F(\cdot)$ such that the value of this data source to the advertisers is positive for values of $v^* > \tau$, and negative for values of $v^* < \tau$.
\end{proof}

To understand this result, note that if a data source helps advertisers distinguish amongst high value advertising opportunities, then this data is likely to help an advertiser identify a high-value advertising opportunity that the advertiser might not have won in the absence of the data.  Such a possibility is beneficial for the advertiser.  But if a data source helps advertisers distinguish amongst low value advertising opportunities, this data is unlikely to help an advertiser much, because advertisers with low values are unlikely to win the auction.  Instead all this data source is likely to do is increase the expected price paid by the winner of the auction.  Thus a data source that helps advertisers distinguish amongst high-value advertising opportunities increases welfare when multiple advertisers have this data, but a data source that only helps advertisers distinguish amongst low-value advertising opportunities decreases welfare.
\section{Working with Multiple Signals} \label{sec:manysignals}

In the previous sections we derived the value of a single data source to the buyer in several diverse settings.  We now address questions related to multiple data sources.  Will additional data sources become more or less valuable when an advertiser already has access to other data sources?  How should an advertiser resolve the trade-off between the cost of a data source and its quality in deciding which of several possible data sources to purchase?

One might think intuitively that when an advertiser is buying multiple independent signals that the marginal value of additional signals would be decreasing in the number of signals that an advertiser has already purchased because each additional signal would do less to refine the advertiser's assessment of the user's type.  However, this need not be the case because the advertiser's value for targeting data depends crucially on the landscape of competing bids, and as a result, a second signal may be much more likely to have an effect on whether the advertiser wants to win the auction than the first signal:

\begin{observation}\label{t:multi-signals}  In the independent private values setting with two types of users, the marginal value of an additional signal need not vary monotonically with the number of signals the advertiser already has access to.

\end{observation}
\begin{proof}
%\textbf{Proof of Theorem \ref{t:multi-signals}}  
Suppose that $v_H = 1$, $v_L = 0$, the prior probability the advertiser assigns to the possibility that the user is type $H$ is $\pi = \frac{1}{2}$, the highest competing bid is always $p = \frac{4}{5}$, and an advertiser is considering purchasing signals of quality $q = \frac{3}{4}$.  Then the marginal value of the first such signal the advertiser purchases is zero because the maximum probability the advertiser will assign to the possibility that the user is type $H$ after the first signal is $\frac{3}{4}$, and the advertiser will thus not bid enough to ever win the auction.  The marginal value of the second such signal the advertiser purchases is strictly positive because if both signals indicate that the user is type $H$, then the advertiser will assign probability $\frac{9}{10}$ to the possibility that the user is type $H$ and the advertiser will be able to obtain a strictly positive payoff by bidding in the auction.  But if the advertiser already has access to a large number of signals, then the marginal value of additional signals is arbitrarily close to zero since the advertiser will be able identify whether the user is type $H$ with virtual certainty even without an additional signal.  Thus the marginal value of an additional signal need not vary monotonically with the number of signals the advertiser has access to. 
\end{proof}

While the value of the data depends on what other signals the advertiser has access to, below we derive a bound on the value solely as a function of the quality of the data, independent of what other data sets are present:

\begin{theorem}\label{t:noprior}
In the independent private value model with two user types, $H$ and $L$, an advertiser's value for a signal with quality $q$ is bounded by $(v_H - v_L)^2 \bar{f} \left(\frac{2}{3}q^3 - \frac{1}{2}q^2 + \frac{1}{24}\right)$, where $\bar{f} \equiv \sup_{p \in [v_L, v_H]} f(p)$. 
\end{theorem}
\begin{proof}
Recall from Lemma~\ref{l:noisy} that the value of the data to the advertiser in this setting is $U_\pi(q) = u_D - u_{ND} =$
$$\int_{\bar{v}}^{v | h} (\pi q (v_H - p) + (1 - \pi)(1 - q)(v_L - p))f(p) dp  -  \int_{v | \ell}^{\bar{v}} (\pi(1 - q)(v_H - p) + (1 - \pi)q (v_L - p)) f(p) dp$$

To derive the bound, we will define a function $T(q)$ such that, $U_\pi(\frac{1}{2})  \leq T(\frac{1}{2})$ and that $\frac{\partial U_\pi}{\partial q} \leq \frac{\partial T}{\partial q}$.  

To proceed, note that:
\begin{align*}
\frac{\partial U_\pi}{\partial q} = &\left(\pi q (v_H - v|h) + (1 - \pi)(1 - q)(v_L - v|h)\right)f(v|h) \frac{\partial v|h}{\partial q} + \int_{\bar{v}}^{v|h} ( \pi(v_H - p) - (1 - \pi)(v_L - p))f(p) dp - \\
&\left(\pi(1 - q)(v_H - v|\ell) + (1 - \pi)q(v_L - v|\ell)\right)f(v|l) \frac{\partial v|\ell}{\partial q} + \int_{v|\ell}^{\bar{v}} (\pi(v_H - p) - (1 - \pi)(v_L - p))f(p) dp \\
=&\int_{v|\ell}^{v|h}(\pi (v_H - p) - (1 - \pi)(v_L - p))f(p) dp, 
\end{align*}

\noindent where we make use of the facts illustrated in previous proofs that $\left(\pi q (v_H - v|h) + (1 - \pi)(1 - q)(v_L - v|h)\right) = 0$ and $\left(\pi(1 - q)(v_H - v|\ell) + (1 - \pi)q(v_L - v|\ell)\right) = 0$.

We can further simplify the derivative:
$$\int_{v|\ell}^{v|h}(\pi (v_H - p) - (1 - \pi)(v_L - p))f(p) dp = (\pi v_H - (1 - \pi)v_L)(F(v|h) - F(v | \ell)) - (2\pi - 1) \int_{v|\ell}^{v|h}p f(p) dp$$

We consider two cases, depending on the sign of $(2\pi - 1)$. When $(2\pi - 1) \geq 0$, we use the fact that $\int_{v|\ell}^{v|h} p f(p) dp \geq (F(v|h) - F(v|\ell)) v | \ell$. Then:
\begin{align*}
\frac{\partial U_\pi}{\partial q} &\leq (\pi v_H - (1 - \pi)v_L)(F(v|h) - F(v | \ell))  - (2\pi - 1) (F(v|h) - F(v|\ell)) v | \ell \\
&=(F(v|h) - F(v| \ell)) \left(\pi v_H - (1 - \pi)v_L - (2\pi-1) v | \ell\right) \\
&= (F(v|h) - F(v | \ell))(v_H - v_L) \cdot \frac{\pi (1 - \pi)}{\pi (1 - q) + (1 - \pi) q}\\
&\leq (F(v|h) - F(v | \ell))(v_H - v_L) q
\end{align*}
The last step follows since for $q \geq \frac{1}{2}$ and $(2\pi - 1) \geq 0$ we have $\frac{\pi (1 - \pi)}{\pi (1 - q) + (1 - \pi) q} \leq q$ for the following reason:  $q(\pi (1 - q) + (1 - \pi) q)$ is minimized on the interval $q \in [\frac{1}{2}, 1]$ when either $q = \frac{1}{2}$ or $q = 1$.  But when $q = \frac{1}{2}$, $q(\pi (1 - q) + (1 - \pi) q) = \frac{1}{2} \geq \pi(1-\pi)$ for all $\pi$.  And when $q = 1$, $q(\pi (1 - q) + (1 - \pi) q) = 1 - \pi \geq \pi(1-\pi)$.  Thus $q(\pi (1 - q) + (1 - \pi) q) \geq \pi (1 - \pi)$ for $q \geq \frac{1}{2}$ and $(2\pi - 1) \geq 0$, and $\frac{\pi (1 - \pi)}{\pi (1 - q) + (1 - \pi) q} \leq q$.

When $(2\pi - 1) \leq 0$, we use the fact that $\int_{v|\ell}^{v|h} p f(p) dp \leq (F(v|h) - F(v|\ell)) v | h$. Then:
\begin{align*}
\frac{\partial U_\pi}{\partial q} &\leq (\pi v_H - (1 - \pi)v_L)(F(v|h) - F(v | \ell))  - (2\pi - 1) (F(v|h) - F(v|\ell)) v | h \\
&=(F(v|h) - F(v| \ell)) \left(\pi v_H - (1 - \pi)v_L - (2\pi-1) v | h\right) \\
&= (F(v|h) - F(v | \ell))(v_H - v_L) \cdot \frac{\pi (1 - \pi)}{\pi q + (1 - \pi)(1 -  q)}\\
&\leq (F(v|h) - F(v | \ell))(v_H - v_L) q
\end{align*}

Here the last step follows since for $q \geq \frac{1}{2}$ and $(2\pi - 1) \leq 0$ we have $\frac{\pi (1 - \pi)}{\pi q + (1 - \pi)(1 - q)} \leq q$ for the following reason:  $q(\pi q + (1 - \pi)(1- q))$ is minimized on the interval $q \in [\frac{1}{2}, 1]$ when either $q = \frac{1}{2}$ or $q = 1$.  But when $q = \frac{1}{2}$, $q(\pi q + (1 - \pi)(1- q)) = \frac{1}{2} \geq \pi(1-\pi)$ for all $\pi$.  And when $q = 1$, $q(\pi q + (1 - \pi)(1 - q)) = \pi \geq \pi(1-\pi)$.  Thus $q(\pi q + (1 - \pi)(1 - q)) \geq \pi (1 - \pi)$ for $q \geq \frac{1}{2}$ and $(2\pi - 1) \leq 0$, and $\frac{\pi (1 - \pi)}{\pi (1 - q) + (1 - \pi) q} \leq q$.

Therefore in both cases: 
\begin{align*}
\frac{\partial U_{\pi}}{\partial q} &\leq (F(v|h) - F(v|\ell))(v_H - v_L) q \\
&\leq \bar{f} (v| h - v|\ell)(v_H - v_L) q\\
&= \bar{f} (v_H - v_L)^2 q (\pi | h -\pi | \ell) \\
&\leq \bar{f} (v_H - v_L)^2 \frac{\pi (1 - \pi) q (2q - 1)}{(\pi q + (1 - \pi)(1 - q))(\pi(1 - q) + (1 - \pi)q)}\\
&\leq \bar{f}(v_H - v_L)^2 q (2q - 1)
\end{align*}

Here the last inequality follows from the fact that $(\pi q + (1 - \pi)(1 - q))(\pi(1 - q) + (1 - \pi)q)$ is minimized on $q \in [\frac{1}{2}, 1]$ when $q = 1$ and $(\pi q + (1 - \pi)(1 - q))(\pi(1 - q) + (1 - \pi)q) = \pi(1 - \pi)$ when $q = 1$.  Thus $(\pi q + (1 - \pi)(1 - q))(\pi(1 - q) + (1 - \pi)q) \geq \pi(1 - \pi)$ for all $q \in [\frac{1}{2}, 1]$ and $\frac{\pi (1 - \pi) q (2q - 1)}{(\pi q + (1 - \pi)(1 - q))(\pi(1 - q) + (1 - \pi)q)} \leq q(2q - 1)$ for all $q \in [\frac{1}{2}, 1]$.

Therefore for any $c$, $T(q)  = \bar{f}(v_H - v_L)^2 (\frac{2}{3}q^3 - \frac{q^2}{2} + c)$ will have grow faster than $U_{\pi}(q)$. To conclude, observe that $U_\pi(\frac{1}{2}) = 0$ for all $\pi$; setting $c = \frac{1}{24}$ ensures that $T(\frac{1}{2}) = U_{\pi}(\frac{1}{2})  = 0$. Combined with the fact that for all $q \geq \frac{1}{2}$ we have: $\frac{\partial U_\pi}{\partial q} \leq \frac{dT}{dq}$, which completes the proof.
\end{proof}

By applying very similar logic we can illustrate bounds on the maximum additional amount that an advertiser would be willing to pay for a more accurate signal regardless of the advertiser's prior.  This is done below:

\begin{corollary}\label{t:noprior2}  Suppose an advertiser has the option of buying two different types of signals with qualities $q_1$ and $q_2$ respectively, where $q_1 > q_2$.  Then if the cost of the first signal is more than $\bar{f}(v_H - v_L)^2[\frac{2}{3}(q_1^3 - q_2^3) - \frac{1}{2}(q_1^2 - q_2^2)]$ greater than the cost of the second signal, the advertiser will always prefer to purchase the second signal regardless of the advertiser's prior.

\end{corollary}
\begin{proof}
We have seen in the proof of the previous theorem that $\frac{\partial U_{\pi}}{\partial q} \leq \bar{f}(v_H - v_L)^2 q (2q - 1)$.  From this it follows that an advertiser's utility for a signal of quality $q_1 > q_2$ cannot exceed the advertiser's utility for a signal of quality $q_2$ by more than $\bar{f}(v_H - v_L)^2[\frac{2}{3}(q_1^3 - q_2^3) - \frac{1}{2}(q_1^2 - q_2^2)]$.
\end{proof}

The results in this section illustrate that there are natural bounds on the price of the data that one can use to quickly rule out whether certain sources of data are cheap enough to be worthwhile, without knowing the advertiser's prior or the entire distribution of competing bids.  If the cost of the data exceeds the bounds in the previous theorems, an advertiser will never want to purchase it.

\section{Reverse Engineering the Value of Data}
\label{sec:reverse}
So far, we have focused on how an advertiser values targeting data.  However, it is also natural to inquire about the value of data from the publisher's perspective.  If a publisher is supplying data that maps the user to a specific population segment, can he  deduce how much an advertiser would value this data solely by observing the advertiser's average bidding behavior for the different realizations of the data?  Unfortunately the answer to this question is no.  We prove this formally below, but first provide some intuition behind the result.  %The advertiser's value for the targeting data will depend critically on how the advertiser values advertising to each of the various types of users, there are various possible ways for the advertiser to value advertising to each of the individual users that are potentially consistent with any given average bidding behavior for the different realizations of the targeting data, and these different possibilities could lead to dramatically different valuations of the targeting data.  In particular, the advertiser's value for this targeting data can range anywhere from zero to some maximal value that represents the value this advertiser would have if this is the only targeting dimension the advertiser could ever possibly care about.  

Suppose the publisher supplies data that identifies whether a user lives in California, and further suppose that the advertiser is making a higher average bid on users from California.  There are several possible ways for the advertiser to exhibit this aggregate bidding behavior.  First, it is possible that the only factor that influences how much the advertiser values a particular advertising opportunity is whether the user lives in California.  In this case, the targeting information that the publisher supplies is truly valuable to the advertiser.

Another possibility is that the advertiser only values advertising to some subset of the users in California that she can already target on, and values all other users equally. In this case, the publisher's data is worthless to the advertiser because the advertiser would make the same bids without this data, yet the average bids for users in and outside of California are the same as before. Thus it is also possible for the advertiser to have zero value for this data.

Furthermore, by taking convex combinations of the two extremes, it is possible to construct scenarios where the advertiser's value for the data can assume any value in this range. We thus obtain the following result:

\begin{theorem}\label{t:reverse}  It is not possible to infer an advertiser's value for targeting data if the publisher only observes the advertiser's average bids for different realizations of the targeting data as well as the distribution of the highest competing bids.  Given these average bids, the advertiser's value for the targeting data may range anywhere from zero to some maximal value.
%In particular, suppose an advertiser is making an average bid of $b_H$ for one realization of the targeting %data, an average bid of $b_L < b_H$ for the other realization of the targeting data, and the different %realizations of the targeting data occur with probabilities $\pi$ and $1 - \pi$ respectively.  Then for any %$v$ between $0$ and $\pi \int_{\bar{b}}^{b_H} (b_H - p) f(p) \; dp + (1 - \pi) \int_{b_L}^{\bar{b}} (p - b_L) %f(p) \; dp$, where $\bar{b} \equiv \pi b_H + (1-\pi)b_L$ and $f(p)$ denotes the pdf corresponding to the %distribution of competing bids, there exists a realization of the advertiser's value for advertising to %individual users that is consistent with this aggregate bidding behavior and would imply that the advertiser's %value for the data is $v$.  
\end{theorem} 
\begin{proof} 
%\textbf{Proof of Theorem \ref{t:reverse}}  
Suppose each user belongs to one of three disjoint sets $T_1, T_2, T_3$. Let $\pi_i$ denote the probability that the user belongs to the set $T_i$. We will use $T_{12}$ to denote the set $T_1 \cup T_2$ and $\pi$ to denote the corresponding probability i.e. $\pi = \pi_1 + \pi_2$. 
Suppose that a publisher is supplying targeting data that indicates whether a user is in the set $T_{12}$ or in the set $T_{3}$. %, where $T_{12} \cup T_{3} = T$, the set of all users.  
Consider an advertiser that has a value $v_i$ for users in the set $T_i$. 
Also suppose the advertiser is making an average bid of $b_{12}$ for users in the set $T_{12}$ and making an average bid of $b_{3} < b_{12}$ for users in the set $T_{3}$. %, where $b_H = b_{12}$ and $b_L = b_3$.  

%Finally suppose that the set $T_{12}$ can be further divided into the disjoint sets $T_{1}$ and $T_{2}$ %satisfying $T_{1} \cup T_{2} = T_{12}$, and the advertiser has a value $v_{1}$ for users in the set $T_{1}$, %$v_{2}$ for users in the set $T_{2}$, and $v_{3}$ for users in the set $T_{3}$.  We assume throughout that %the probability a user is in the set $T_{1}$ is $\pi_{1}$, the probability a user is in the set $T_{2}$ is %$\pi_{2}$, and the probability a user is in the set $T_{3}$ is $\pi_{3}$, where $\pi_1 + \pi_2 = \pi$.
%
First suppose that $v_{1} = v_{2} > v_{3}$ and the advertiser has no additional targeting data.  Note that in this scenario, it must be the case that $b_{12} = v_{1} = v_{2}$ and $b_{3} = v_{3}$ and if the advertiser did not have access to the targeting data, then the advertiser would make a bid of $b = \pi_{1} v_{1} + \pi_{2} v_{2} + \pi_{3} v_{3}$.  Thus the advertiser's value for the targeting data is 
\begin{equation} \label{vod}
(\pi_{1} + \pi_{2}) \int_{b}^{b_{12}} (b_{12} - p) f(p) \; dp + \pi_{3} \int_{b_{3}}^{b} (p - b_{3}) f(p) \; dp 
%%&= \pi \int_{\bar{b}}^{b_H} (b_H - p) f(p) \; dp + (1 - \pi) \int_{b_L}^{\bar{b}} (p - b_L) f(p) \; dp
\end{equation}
We now seek to show that the advertiser's value for the targeting data can range anywhere from zero to the value given in the equation above.  To see this, let $S_{1}$ denote a set of random users drawn from the set $T_{1}$ that compose a fraction $\alpha$ of the users in $T_{1}$.  Suppose the advertiser has access to targeting data which identifies whether a user is in the set $S_{1}$ or not in the set $S_{1}$.  In this case, when the advertiser also has access to the data that indicates whether a user is in the set $T_{12}$ or in the set $T_{3}$, the advertiser will use the following bidding strategy:  If the advertiser knows the user is in the set $S_{1}$, then the advertiser makes a bid of $b_{1} = v_{1}$.  If the advertiser knows the user is in the set $T_{12}$ but is not in the set $S_{1}$, then the advertiser makes a bid of $b_{2}(\alpha) = \frac{(1-\alpha)\pi_{1}v_{1} + \pi_{2} v_{2}}{(1-\alpha)\pi_{1} + \pi_{2}}$.  And if the advertiser knows the user is in the set $T_{3}$, then the advertiser makes a bid of $b_{3} = v_{3}$.  From this it follows that an advertiser's average bid for users in the set $T_{12}$ is $\frac{\alpha \pi_{1} b_{1} + ((1-\alpha)\pi_{1} + \pi_{2}) b_{2}(\alpha)}{\pi_{1} + \pi_{2}} = \frac{\alpha \pi_{1} v_{1} + (1-\alpha)\pi_{1}v_{1} + \pi_{2} v_{2}}{\pi_{1} + \pi_{2}} = \frac{\pi_{1} v_{1} + \pi_{2} v_{2}}{\pi_{1} + \pi_{2}}$, so $b_{12} = \frac{\pi_{1} v_{1} + \pi_{2} v_{2}}{\pi_{1} + \pi_{2}}$.

Now let $v_{1}$ and $v_{2}$ vary with $\alpha$ in the following way:  We let $v_{1}(\alpha) \equiv b_{12} + \alpha \frac{\pi_2}{\pi_1}(b_{12} - b_{3})$ and let $v_{2}(\alpha) \equiv (1-\alpha)b_{12} + \alpha b_{3}$.  Note that these values of $v_{1}(\alpha)$ and $v_{2}(\alpha)$ ensure that for any value of $\alpha$, it is necessarily the case that $b_{12} = \frac{\pi_{1} v_{1}(\alpha) + \pi_{2} v_{2}(\alpha)}{\pi_{1} + \pi_{2}}$.  Further note that we can calculate the advertiser's value for the targeting data being supplied by the publisher when the advertiser's values vary with $\alpha$ in this way.  If the advertiser did not have access to the publisher's targeting data, then the advertiser would still bid $b_{1} = v_{1}(\alpha)$ when the advertiser knows the user is in the set $S_{1}$, but the advertiser would now bid $\hat{b}(\alpha) = \frac{((1-\alpha)\pi_{1} + \pi_{2})b_{2}(\alpha) + \pi_{3} b_{3}}{(1-\alpha)\pi_{1} + \pi_{2} + \pi_{3}}$ if the user is not in the set $S_{1}$ rather than making different bids on occasions where the user is in the set $T_{12}$ rather than $T_{3}$.  Thus the advertiser's value for the targeting data being supplied by the publisher is

\begin{equation*}
((1-\alpha)\pi_{1} + \pi_{2}) \int_{\hat{b}(\alpha)}^{b_{2}(\alpha)} (b_{2}(\alpha) - p) f(p) \; dp + \pi_{3} \int_{b_{3}}^{\hat{b}(\alpha)} (p - b_{3}) f(p) \; dp 
\end{equation*}

Now note that when $\alpha = 0$, this expression is equal to the value of data given in equation (\ref{vod}) that gives the advertiser's value for the targeting data when $v_{1} = v_{2} > v_{3}$ and the advertiser has no additional targeting data.  Also note that when $\alpha = 1$, this expression equals $0$.  Furthermore, this expression is continuous in $\alpha$, so the expression assumes every possible value from $0$ to the value in equation (\ref{vod}) when $\alpha$ ranges from $0$ to $1$.  From this it follows that the advertiser's value for the targeting data can range anywhere from zero to the value given in equation (\ref{vod}). 
\end{proof}

\section{Conclusion}
In this paper we extensively analyzed the problem of valuing targeting data to an advertiser, and showed precisely how it depends both on the advertiser's parameters (value, budget, etc.)  and the actions of other players (competing bids, participation, etc.). A natural next step is to design truthful and welfare maximizing mechanisms that can aid in a creation of an effective  marketplace. 

\bibliography{data}

\begin{thebibliography}{10}

\bibitem{Abraham13}
I.~Abraham, S.~Athey, M.~Babaioff, and M.~Grubb.
\newblock Peaches, lemons, and cookies: Designing auction markets with
  dispersed information.
\newblock Technical report, Microsoft Research, May 2013.

\bibitem{Babaioff12}
M.~Babaioff, R.~Kleinberg, and R.~P. Leme.
\newblock Optimal mechanisms for selling information.
\newblock In Faltings et~al. \cite{DBLP:conf/sigecom/2012}, pages 92--109.

\bibitem{Bergemann11}
D.~Bergemann and A.~Bonatti.
\newblock Targeting in advertising markets: Implications for offline vs. online
  media.
\newblock {\em RAND Journal of Economics}, 42(3):414--443, 2011.

\bibitem{AdaBid2}
V.~Bharadwaj, P.~Chen, W.~Ma, C.~Nagarajan, J.~Tomlin, S.~Vassilvitskii,
  E.~Vee, and J.~Yang.
\newblock Shale: an efficient algorithm for allocation of guaranteed display
  advertising.
\newblock In Q.~Yang, D.~Agarwal, and J.~Pei, editors, {\em KDD}, pages
  1195--1203. ACM, 2012.

\bibitem{Chen}
Y.~Chen, P.~Berkhin, B.~Anderson, and N.~R. Devanur.
\newblock Real-time bidding algorithms for performance-based display ad
  allocation.
\newblock In C.~Apt{\'e}, J.~Ghosh, and P.~Smyth, editors, {\em KDD}, pages
  1307--1315. ACM, 2011.

\bibitem{Emek12}
Y.~Emek, M.~Feldman, I.~Gamzu, R.~P. Leme, and M.~Tennenholtz.
\newblock Signaling schemes for revenue maximization.
\newblock In Faltings et~al. \cite{DBLP:conf/sigecom/2012}, pages 514--531.

\bibitem{DBLP:conf/sigecom/2012}
B.~Faltings, K.~Leyton-Brown, and P.~Ipeirotis, editors.
\newblock {\em ACM Conference on Electronic Commerce, EC '12, Valencia, Spain,
  June 4-8, 2012}. ACM, 2012.

\bibitem{Fu12}
H.~Fu, P.~R. Jordan, M.~Mahdian, U.~Nadav, I.~Talgam-Cohen, and
  S.~Vassilvitskii.
\newblock Ad auctions with data.
\newblock In M.~Serna, editor, {\em SAGT}, volume 7615 of {\em Lecture Notes in
  Computer Science}, pages 168--179. Springer, 2012.

\bibitem{Ghosh12}
A.~Ghosh, M.~Mahdian, R.~P. McAfee, and S.~Vassilvitskii.
\newblock To match or not to match: Economics of cookie matching in online
  advertising.
\newblock In Faltings et~al. \cite{DBLP:conf/sigecom/2012}, pages 741--753.

\bibitem{RepBid}
A.~Ghosh, R.~P. McAfee, K.~Papineni, and S.~Vassilvitskii.
\newblock Bidding for representative allocations for display advertising.
\newblock In S.~Leonardi, editor, {\em WINE}, volume 5929 of {\em Lecture Notes
  in Computer Science}, pages 208--219. Springer, 2009.

\bibitem{Ghosh07}
A.~Ghosh, H.~Nazerzadeh, and M.~Sundararajan.
\newblock Computing optimal bundles for sponsored search.
\newblock In X.~Deng and F.~C. Graham, editors, {\em WINE}, volume 4858 of {\em
  Lecture Notes in Computer Science}, pages 576--583. Springer, 2007.

\bibitem{AdaBid}
A.~Ghosh, B.~I.~P. Rubinstein, S.~Vassilvitskii, and M.~Zinkevich.
\newblock Adaptive bidding for display advertising.
\newblock In J.~Quemada, G.~Le{\'o}n, Y.~S. Maarek, and W.~Nejdl, editors, {\em
  WWW}, pages 251--260. ACM, 2009.

\bibitem{Hummel13}
P.~Hummel and R.~P. McAfee.
\newblock When does improved targeting increase revenue?
\newblock Technical report, Google Inc., April 2014.

\bibitem{Milgrom82}
P.~R. Milgrom and R.~J. Weber.
\newblock A theory of auctions and competitive bidding.
\newblock {\em Econometrica}, 50(5):1089--1122, 1982.

\end{thebibliography}

\end{document}